\documentclass[11pt,reqno]{amsart} % 'reqno' = right‑aligned equation numbers

\usepackage[margin=1.1in]{geometry}
\usepackage{afterpage,placeins,float} 
\usepackage[foot]{amsaddr} 
\usepackage[T1]{fontenc}
\usepackage{lmodern} % Latin Modern = clean Computer Modern look
\usepackage{microtype} % better kerning / protrusion
\usepackage{booktabs} % gives \toprule, \midrule, \bottomrule
\usepackage{siunitx} % decimal-aligned numbers and the \sisetup{} command
\usepackage{pgfplots} % loads TikZ automatically
\usepackage{url}
\usepackage{subcaption} 
\usepackage{booktabs} 

\usepackage{amsmath,amssymb,amsthm,mathtools}
\numberwithin{equation}{section} % (1.1), (1.2), …

\theoremstyle{plain}

\DeclareMathOperator{\Var}{Var}
\newcommand{\E}{\mathbb{E}}
\newcommand{\Prob}{\mathbb{P}}

\usepackage{enumitem}
\setlist[itemize]{leftmargin=1.5em}
\setlist[enumerate]{leftmargin=1.5em}

\usepackage{booktabs}
\usepackage{graphicx}
\usepackage{subcaption}

\usepackage[]{amsmath,amssymb,amsthm}
\usepackage{graphicx}
\usepackage{xcolor}
\newtheorem{claiminner}{Claim}

\newtheorem{lemma}{Lemma}[section] % Lemmas numbered within each section
\newtheorem{definition}{Definition}[section]
\newtheorem{corollary}{Corollary}[section] 
\theoremstyle{remark}
\newtheorem*{remark}{Remark}

\title{Loss-Versus-Rebalancing under Deterministic and Generalized block-times}

\author[A.~Nezlobin]{Alex Nezlobin}
\author[M.~Tassy]{Martin Tassy}

\email{%
	\texttt{NezlobinAlexander@gmail.com},
	\texttt{mtassy@wilkinsrc.com}
}

\date{\today} % You can replace \today with a specific date if needed

\begin{document}
	
	\maketitle
	
\begin{abstract}
Although modern blockchains almost universally produce blocks at fixed intervals, existing models still lack an analytical formula for the loss-versus-rebalancing (LVR) incurred by Automated Market Makers (AMMs) liquidity providers in this setting. Leveraging tools from random walk theory, we derive the following closed-form approximation for the per block per unit of liquidity expected LVR under constant block time:

	\[
	\overline{\mathrm{ARB}}=
	\frac{\,\sigma_b^{2}}
	{\,2+\sqrt{2\pi}\,\gamma/(|\zeta(1/2)|\,\sigma_b)\,}+O\!\bigl(e^{-\mathrm{const}\frac{\gamma}{\sigma_b}}\bigr)\;\approx\; 
	\frac{\sigma_b^2}{\,2 + 1.7164\,\gamma/\sigma_b},
	\]
	where $\sigma_b$ is the intra-block asset volatility, $\gamma$ the AMM spread and $\zeta$ the Riemann Zeta function. Our large Monte Carlo simulations show that this formula is in fact quasi-exact across practical parameter ranges.
	
Extending our analysis to arbitrary block-time distributions as well, we demonstrate both that—under every admissible inter-block law—the probability that a block carries an arbitrage trade converges to a universal limit, and that only constant block spacing attains the asymptotically minimal LVR. This shows that constant block intervals provide the best possible protection against arbitrage for liquidity providers.

\end{abstract}

	\section{Introduction}
	
	Automated market makers (AMMs) have emerged as a cornerstone of the decentralized finance (DeFi) trading landscape, enabling assets to be exchanged without traditional order-book mechanisms. By aggregating liquidity from external liquidity providers (LPs) and using predetermined pricing formulas \cite{hanson2007logarithmic, othman2013automated, adams2020uniswap}, AMMs democratize market making. However, this convenience comes with inherent risks: LPs face an adverse selection cost due to their passive quoting of prices that can become stale between block confirmations. Arbitrageurs exploit price discrepancies between the AMM and external markets, resulting in a continuous loss of value from LP positions \cite{angeris2019analysis, capponi2025}. This phenomenon driven by the discrete timing of blockchain transactions has been recognized in practice as \emph{divergence loss} or \emph{impermanent loss}, reflecting the opportunity cost of providing liquidity instead of simply holding the assets \cite{angeris2019analysis}. 
	
	While impermanent loss captures the general underperformance of an LP relative to holding, it conflates multiple sources of risk and thus lacks specificity in attributing losses to stale pricing. To isolate the loss incurred purely from arbitrage on stale quotes, Milionis \emph{et al.} \cite{Milionis2022} introduced the concept of \emph{loss-versus-rebalancing} (LVR). LVR is defined as the shortfall of an LP’s portfolio value relative to a continuously rebalanced portfolio that tracks the AMM’s asset ratio, effectively quantifying the cost of offering liquidity due to arbitrageurs exploiting outdated quotes. Subsequent work by Milionis \emph{et al.} \cite{Milionis2023-poisson} derived closed-form expressions for LVR in the idealized case of Poisson-distributed block arrival times, corresponding to proof-of-work blockchains (e.g. Ethereum pre-merge). 
	
% ------------------------------------------------------------------
% Polished text — preserve user-supplied structure, elevate register
% ------------------------------------------------------------------

	However, modern blockchain networks are overwhelmingly based on proof-of-stake or other consensus mechanisms that produce blocks at fixed intervals (for example, Ethereum post-merge or Solana), rather than an exponential timing. Existing analytical results for LVR are not directly applicable to these settings, leaving a significant gap in our understanding of how the block-time distribution influences LVR. Filling this gap is crucial: as on-chain market making becomes more competitive, even modest reductions in LVR can substantially improve LP profitability and the overall efficiency of AMMs.

One of the main goals of this paper is to develop the mathematical methodology to address this shortfall. Continuous diffusions and Itô calculus, the traditional tools for AMM analysis, are ill-suited for non-exponential block times. While diffusion models capture certain macro properties, they obscure the sharp boundary effects of discrete block production, leading to inaccurate micro-level LVR estimates. On a blockchain, price evolution is better modeled as a discrete random walk, requiring a new mathematical framework to capture its granular behavior.

We address this gap by developing a discrete-time, continuous-state Markov chain model of an AMM that supports arbitrary block-time distributions, with a focus on the deterministic block-time case. Evolving block-by-block, the model reflects AMM state dynamics and provides tractable expressions for LP losses.
Analytically, our results leverage the vast probabilistic apparatus of random walks on strips. This framework was initially developed in queueing theory (\cite{Kendall1951,Kingman1961}) and later refined by Spitzer’s fluctuation theory \cite{Spitzer1956}, ladder-height techniques (\cite{ChangPeres1997,nagaev2010exact,fuh2007asymptotic}), and boundary crossing results for strips (\cite{Lotov1996,lotov2023exact,khaniyev2004asymptotic,Khaniyev2022})
By transplanting these tools to the AMM context, we derive analytic LVR formulas under general block-time laws and, in the uniformly spaced case, obtain a closed-form expression whose error term decays exponentially with the intra-block volatility of the asset. Monte-Carlo experiments confirm that the resulting approximation is
remarkably accurate, thereby extending the theoretical foundation of
LVR to the contemporary blockchain systems.

	\emph{Summary of contributions.} Our main findings are summarized as follows:
	\begin{enumerate}
		\item We establish a general decomposition of LP arbitrage losses that holds for any block-time distribution, generalizing a structure which was first observed by Milionis \emph{et al.} (2023) in the special case of Poisson-distributed block times \cite{Milionis2023-poisson}. In particular, we show that the expected arbitrage loss per block can be factored as $\overline{ARB} = P_{\mathrm{trade}} \times \overline{LVR}$, where $P_{\mathrm{trade}}$ is the probability that an arbitrage trade occurs in a given block and $\overline{LVR}$ is the expected loss conditional on an arbitrage trade occurring.
		
		\item For uniformly spaced blocks, we obtain an analytic expression for $\overline{LVR}$ in the small-volatility regime, accurate up to an exponentially small error term. This formula is practically exact for typical market parameters (yielding error $<0.01\%$ in simulations). These results are presented in Table~\ref{table2}. Moreover, our results indicate that, all else equal, moving from a Poisson block-time to a constant block-time reduces the per-block LVR by up to $17.4\%$ in the fast-block regime (see Figure~\ref{fig:poissonuniform} and Figure~\ref{fig:rho_diff_plot}).
		
			\item We find that the asymptotic arbitrage probability $P_{\mathrm{trade}}$ (in the limit of small per-block price volatility) is universal, i.e., to first order 
		it does not depend on the block-time distribution. We derive an explicit formula for $P_{\mathrm{trade}}$ in this regime, revealing a fundamental invariance in arbitrage frequency across different block arrival processes. In contrast, the magnitude of the loss per arbitrage event, $\overline{LVR}$, does depend on the block-time distribution. As an important implication, we show that among all block-time distributions with a given mean, the constant distribution uniquely minimizes the asymptotic LVR. In other words, there is no better choice of the block distribution that could reduce the adverse selection cost for LPs.
		\end{enumerate}
		
\begin{table}[htbp]
	\centering
	\renewcommand{\arraystretch}{1.45}
	\caption{Asymptotic expressions for $P_{\mathrm{trade}}$, $\overline{LVR}$,
		and $\overline{ARB}$ under different block-time distributions, with
		the leading-order error term shown separately.}
	\label{tab:results_summary}
	\begin{tabular}{lcccc}
		\toprule
		& $P_{\mathrm{trade}}$ & $\overline{LVR}$ & $\overline{ARB}$ & Error \\
					\addlinespace[2pt]
		\midrule
				\addlinespace[4pt]
		%--------------------------- Poisson -------------------------------
		
		Poisson &
		$\displaystyle \bigl[1+\gamma/(\sqrt{2}\sigma_b)\bigr]^{-1}$ &
		$\displaystyle \frac{\ell\sigma_b^{2}}{2}$ &
		$\displaystyle \frac{\ell\sigma_b^{2}}
		{\,2+\sqrt{2}\gamma/\sigma_b\,}$ &
		$0$ \\[2pt]
		\addlinespace[7pt]
		%--------------------------- Uniform -------------------------------
		Uniform &
		$\displaystyle
		\bigl[\gamma/(\sqrt{2}\sigma_b)+|\zeta(1/2)|/\sqrt{\pi}\bigr]^{-1}$ &
		$\displaystyle
		\frac{\ell\sigma_b^{2}\,|\zeta(1/2)|}{2\sqrt{\pi}}$ &
		$\displaystyle
		\frac{\ell\sigma_b^{2}}
		{\,2+\sqrt{2\pi}\gamma/(|\zeta(1/2)|\sigma_b)\,}$ &
		$O\!\bigl(e^{-c\frac{\gamma}{\sigma_b}}\bigr)$ \\[2pt]
		\addlinespace[9pt]
		%--------------------------- General -------------------------------
		General &
		$\displaystyle \sqrt{2}\sigma_b/\gamma$ &
		$\frac{\ell\sigma_b^{2}}{2}\!\bigl(\frac{|\zeta(1/2)|}{\sqrt{\pi}}+C_\mu\bigr)$ &
		$\frac{\ell\sigma_b^{3}}{\sqrt{2}\gamma}\left(\frac{|\zeta(1/2)|}{\sqrt{\pi}}+C_{\mu}\right)$ &
		$O\!\bigl(\sigma_b/\gamma\bigr)$ \\
		\bottomrule
	\end{tabular}
\end{table}

	Table~\ref{tab:results_summary} summarizes these analytical results for the representative cases of Poisson and constant block-time, as well as for the general asymptotic case. In this table, $\sigma_b$ denotes the volatility of the asset over a block, $\gamma$ denotes the AMM’s internal spread parameter, $\ell$ denotes the liquidity provided per percentage point of price change, $\zeta(\cdot)$ denotes the Riemann zeta function and $C_\mu$ is a non-negative constant depending on the reference block-distribution.

	The remainder of the paper is organized as follows. \textbf{Section~2} introduces our discrete-time Markov chain model and derives general analytical results for LVR, forming the foundation for the subsequent analysis. \textbf{Section~3} focuses on the constant block-time case (characteristic of  proof-of-stake blockchains), where we derive closed-form expressions for key quantities such as the arbitrage probability and the expected LVR. \textbf{Section~4} extends the analysis to arbitrary block-time distributions, establishing distribution-invariant properties and showing that the constant block-time distribution uniquely minimizes asymptotic LVR among all block-time distributions with a given mean. Finally, we close with brief remarks and outline several directions for future research.

	% ---------------------------------------------------------------
	% Float-only page: TABLE 2 (top) + FIGURE 1 (bottom)
	% ---------------------------------------------------------------
	\afterpage{%
%		\clearpage % start a fresh page and flush earlier floats
		
		% ======================= TABLE 2 ===========================
		\begin{table}[H]
			\centering
			\caption{Relative deviation (in percent) between $10^{9}$-path
				Monte-Carlo estimates (Std. error $\approx$ 0.01 \%) and the theoretical expression for $\rho_b=\frac{\gamma}{\sigma_b}$ smaller than $5$. Here $\Delta_{P_{\mathrm{trade}}}$ is the error for $P_{\mathrm{trade}}$, $\Delta_{\overline{LVR}}$ for $\overline{LVR}$, and $\Delta_{\overline{ARB}}$ for the ratio $\overline{ARB}$. 
				The data illustrates that the theoretical and Monte Carlo values of $\overline{ARB}$ become indistinguishable for $\sigma_b \le \gamma$.}
			\label{table2}
			\sisetup{
				table-number-alignment = center,
				round-mode = places,
				round-precision = 3,
				table-format = +1.3
			}
			\begin{tabular}{
					S[table-format=2.1]
					S
					S
					S
				}
				\toprule
				{$\rho_b$} &
				{$\Delta_{P_{\mathrm{trade}}}$} &
				{$\Delta_{\overline{LVR}}$} &
				{$\Delta_{\overline{ARB}}$} \\[-3pt]
				& \multicolumn{3}{c}{}\\
				\midrule
				0.5 & 5.138 & 5.657 & 0.493 \\
				0.7 & 2.568 & 2.823 & 0.248 \\
				0.8 & 1.722 & 1.882 & 0.158 \\
				0.9 & 1.108 & 1.179 & 0.070 \\
				1.0 & 0.654 & 0.668 & 0.014 \\
				2.0 & -0.076 & -0.093 & -0.017 \\
				3.0 & 0.003 & 0.011 & 0.009 \\
				4.0 & 0.005 & 0.005 & 0.000 \\
				5.0 & 0.003 & 0.008 & 0.005 \\
				\bottomrule
			\end{tabular}
		\end{table}
		\vspace{2cm}
		% ======================= FIGURE 1 ==========================
		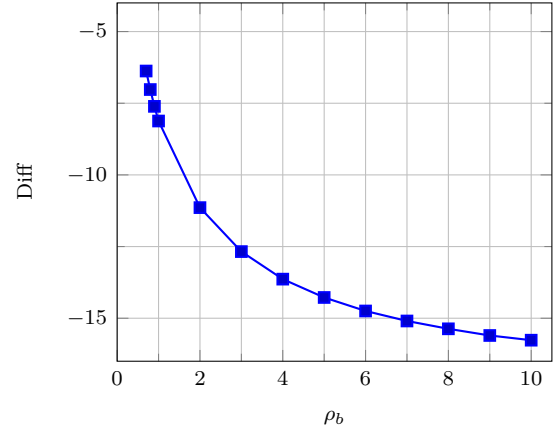
\begin{figure}[H]
			
			\centering
			
			% ---------- LEFT sub-table ----------
			\begin{subfigure}[t]{.46\linewidth}
				\centering
				\captionsetup{justification=centering}
				\caption{Relative difference between simulation and\\
					Poisson-block formula \\[-2pt]
					$(\mathrm{ARB}_{\mathrm{sim}}-\mathrm{ARB}_{\mathrm{Pois}})/\mathrm{ARB}_{\mathrm{Pois}}\;$[\%]}
				\label{fig:poissonuniform}
				{%
					\sisetup{table-format=+2.3}
				\begin{tabular}{S[table-format=2.1] @{\hspace{1.5em}} S}
					\toprule
					{$\rho_b$} & {Diff\%} \\
					\midrule
					0.5  & -4.820  \\ 
					0.7  & -6.376  \\ 
					0.8  & -7.022  \\ 
					0.9  & -7.609  \\ 
					1.0  & -8.120  \\ 
					2.0  & -11.141 \\ 
					3.0  & -12.675 \\ 
					4.0  & -13.636 \\ 
					5.0  & -14.276 \\ 
					6.0  & -14.747 \\ 
					7.0  & -15.096 \\ 
					8.0  & -15.369 \\ 
					9.0  & -15.602 \\ 
					10.0 & -15.766 \\ 
					\bottomrule
				\end{tabular}
				}
			\end{subfigure}
			\hfill
			% ---------- RIGHT sub-plot ----------
			\begin{subfigure}[t]{.46\linewidth}
				\centering
				\captionsetup{justification=centering}
				\caption{Plot of the Constant vs.\ Poisson $\mathrm{ARB}$ improvement\\[-2pt]
					as a function of $\rho_b$ (in \%)}
				\label{fig:rho_diff_plot}
				\begin{tikzpicture}
					\begin{axis}[
						width=\linewidth,
						xmin=0, xmax=10.5,
						ymin=-16.5, ymax=-4,
						xlabel={$\rho_b$},
						ylabel={Diff},
						grid=both,
						minor tick num=1,
						tick label style={font=\scriptsize},
						label style={font=\scriptsize},
						]
						\addplot+[mark=square*, mark size=2pt, thick] table[row sep=\\]{
							n diff
							0.5 -4.820\\
							0.7 -6.376\\
							0.8 -7.022\\
							0.9 -7.609\\
							1.0 -8.120\\
							2.0 -11.141\\
							3.0 -12.675\\
							4.0 -13.636\\
							5.0 -14.276\\
							6.0 -14.747\\
							7.0 -15.096\\
							8.0 -15.369\\
							9.0 -15.602\\
							10.0 -15.766\\
						};
					\end{axis}
				\end{tikzpicture}
			\end{subfigure}
			
			\caption{Deviation of the Constant-Block $\mathrm{ARB}$ from the Poisson-blocks benchmark as a function of $\rho_b$. When $\rho_b$ decreases, the percentage gain increases towards an asymptotic limit of $\approx 17.4 \%$.}
			\label{fig:rho_table_and_plot}
		\end{figure}
		
		\clearpage % ensure subsequent text starts on a new page
	}

	\section{Markov Chain Model}
	
	 We now present the discrete-time, continuous-state Markov chain model that will be used to analyze the impact of the block-time distribution on LVR. Our model offers a tractable alternative to the continuous-time diffusion models typically employed, while still capturing the essential system dynamics.

	\subsection{Model Specification and Dynamics of the Markov Chain}

	We consider a blockchain where block generation follows a random block-time distribution. Let $\mu$ be a base distribution on $[0,\infty)$ with finite first and second moments (we assume $\mathbb{E}[\mu]=1$ for normalization). For a given average block-time $t>0$, denote by $\mu_t$ the distribution of $tX$ when $X\sim \mu$ (so that $\mu_t$ has mean $t$). Accordingly, we model the sequence of block-times as an i.i.d. sequence $(U_i)_{i\in\mathbb{N}}$ with $U_i \sim \mu_t$ for each $i$. One particular case of interest is the \emph{constant block-time model}, where $\mu_t$ is a Dirac delta at $t$ (i.e., each block has exactly length $t$). This deterministic block-time scenario, typical of proof-of-stake systems~\cite{Buterin2020}, will be a focal case in our analysis.
	
	We assume the asset price $S_t$ follows a geometric Brownian motion with zero drift and volatility $\sigma$. In other words, over an interval of length $\Delta$, the \emph{log-price} change $\ln(S_{t+\Delta}/S_t)$ is $\mathcal{N}(0,\sigma^2\Delta)$. (We restrict attention to the driftless case; incorporating a non-zero drift for $S_t$ is left for future work.) 
	
	The AMM is assumed to have a constant \emph{internal spread} $\gamma$, and the liquidity provider supplies liquidity with constant density $\ell$ (USD per percentage point of price change). In other words, an arbitrageur must trade $\ell \cdot p$ USD worth of asset to shift the AMM’s quoted price by $p\%$. 
	
	Our Markov chain model tracks the log-price at block boundaries relative to two fixed arbitrage thresholds. We set the initial log-price (at the start of block 1) to 0 without loss of generality. We then define the no-arbitrage region to be the interval $[0,\gamma]$ in log-price space. If the log-price remains between 0 and $\gamma$, no arbitrage occurs; if it moves outside this interval (either above $\gamma$ or below 0), an arbitrage opportunity is realized. Moreover, conditional on a block of length $U$, the log-price change in that block is $\Delta \ln S \sim \mathcal{N}(0,\sigma^2 U)$. It is convenient to introduce a time-scaled volatility $\sigma_b := \sigma\sqrt{t}$ (the standard deviation of log-price changes over a block of average length $t$). Then $\Delta \ln S$ has variance $\sigma_b^2 \cdot (U/t)$, and in the special case of deterministic block length $t$ we simply have $\Delta \ln S \sim \mathcal{N}(0,\sigma_b^2)$. Another reason why we prefer adopting the notation $\sigma_b$ rather than the usual $\sigma\sqrt{t}$  is that all our results will be independent of the exact way the asset variance is scaling as long as the block distribution can be written as $\sigma_b\mu$.
	
	To simplify the analysis, we normalize the price increments and the arbitrage bounds. Define the normalized log-return $X_i := \Delta \ln S_i / \sigma_b$ for block $i$, and replace $\gamma$ by $\rho_b := \gamma/\sigma_b$. Under an average block time $t$, each $X_i$ can be treated as $\mathcal{N}(0,\mu)$ (i.i.d.), so the log-price movement per block is measured in units of $\sigma_b$. This rescaling allows us to work with a fixed interval $[0,\rho_b]$ for the log-price and standardizes the variance of the jumps $X_i$. An important consequence of our setup is that, since liquidity per basis point is constant, we can recenter the log-price after each arbitrage event. In particular, whenever an arbitrage occurs, we reset the reference log-price to 0 at the start of the next block (and keep the arbitrage bounds $0$ and $\rho_b$ unchanged). This implies that we only need to track cumulative percentage losses over time, as the process between arbitrage events has identical dynamics in this relative frame.
	
	We now formalize the Markov chain. The state at the end of block $n$ is given by the tuple $(M_n,\; LVR_n,\; ARB_n)$, where: 
	\begin{itemize}
		\item $M_n$ is the \emph{relative log-price position} within the interval $[0,\rho_b]$ at the end of block $n$. (Here $M_n=0$ or $M_n=\rho_b$ indicates that the log-price is exactly at an arbitrage threshold, while $0<M_n<\rho_b$ means it lies strictly inside the no-arbitrage region.)
		\item $LVR_n$ is the \emph{Loss-Versus-Rebalancing} incurred by the LP during block $n$ (measured in USD). By definition, $LVR_n=0$ if no arbitrage occurred in block $n$, and $LVR_n>0$ if an arbitrage took place.
		\item $ARB_n$ is the \emph{cumulative LVR} up to and including block $n$, i.e. the total loss in USD the LP has suffered from all arbitrage events up to block $n$. We have $ARB_0=0$ at inception.
	\end{itemize}
	
	State transitions from block $n$ to $n+1$ are governed by whether an arbitrage is triggered in block $n+1$. Denote by $X_{n+1}$ the normalized log return in block $n+1$ (as defined above):
	\begin{enumerate}
		\item \emph{No arbitrage:} If $M_n + X_{n+1} \in [0,\rho_b]$, then the log-price remains within the no-arbitrage bounds during block $n+1$. In this case, no new loss is incurred: 
		\[
		M_{n+1} = M_n + X_{n+1}, \qquad 
		LVR_{n+1} = 0, \qquad 
		ARB_{n+1} = ARB_n~.
		\]
		\item \emph{Arbitrage event:} If $M_n + X_{n+1} \notin [0,\rho_b]$, then the log-price crosses one of the thresholds in block $n+1$, signaling an arbitrage opportunity. Let $b$ be the boundary of $[0,\rho_b]$ that is exceeded by $M_n + X_{n+1}$ (so $b=0$ if $M_n + X_{n+1}<0$, or $b=\rho_b$ if $M_n + X_{n+1} > \rho_b$). 
		We define the \emph{overshoot} beyond the boundary as $|\,M_n + X_{n+1} - b\,|$. Assuming a constant liquidity $\ell$, an arbitrage triggered by a price deviation $d$ (the difference between the AMM's quoted price and the external price) incurs a cost to the LP of $\int_0^d \ell\,u\,du = \frac{\ell d^2}{2}$. Therefore the LVR incurred in block $n+1$ can be expressed as 
		\[
		LVR_{n+1} = \frac{\ell}{2}\left(\,M_n + X_{n+1} - b\,\right)^2,
		\] 
		 After this arbitrage, we reset the state for the next block by setting 
		\[
		M_{n+1} = 0, \qquad ARB_{n+1} = ARB_n + LVR_{n+1}~,
		\] 
		i.e. the new relative log-price position starts again at 0, and the cumulative loss is updated.
	\end{enumerate}
	
\subsection{Arbitrage Probability and Average LVR Decomposition}

Leveraging standard concentration results for Markov chains. We now propose a simple but useful decomposition of the LVR.

Consider the asymptotic loss per block (long-run loss per block experienced by the LPs):
\[ 
\overline{\mathrm{ARB}} \;=\; \lim_{N \to \infty} \frac{ARB_N}{N}\,. 
\] 
Here $ARB_N$ denotes the total accumulated LVR (loss) up to block $N$. Let $\theta_N$ be the number of arbitrage trades that occur by the end of block $N$, and let $\tau_N$ be the index of the block in which the $N$-th arbitrage takes place. Since arbitrage losses occur only at the blocks $\tau_i$, the total loss by time $N$ can be expressed as 
\[ 
ARB_N \;=\; \sum_{i=1}^{\theta_N} LVR_{\tau_i}\,, 
\] 
where $LVR_{\tau_i}$ is the loss incurred at the $i$-th arbitrage event. It follows that 
\[ 
\frac{ARB_N}{N} \;=\; \frac{\theta_N}{N} \cdot \frac{1}{\theta_N}\sum_{i=1}^{\theta_N} LVR_{\tau_i}\,. 
\]

As $N \to \infty$, provided that arbitrage opportunities continue indefinitely (so that $\theta_N \to \infty$ almost surely), we can apply the Strong Law of Large Numbers (SLLN) to each factor on the right-hand side. In particular, this gives, almost surely,
\[ 
\lim_{N \to \infty}\frac{\theta_N}{N} \;=\; \frac{1}{\mathbb{E}[\tau_1]}\,,
\] 
and 
\[ 
\lim_{N \to \infty}\frac{1}{\theta_N}\sum_{i=1}^{\theta_N} LVR_{\tau_i} \;=\; \mathbb{E}\!\big[\,LVR_{\tau_1}\big]\,.
\] 
Combining these two limits, we obtain the asymptotic decomposition 
\begin{equation}\label{eq:decomp}
	\overline{\mathrm{ARB}} \;=\; P_{\text{trade}} \times \overline{\mathrm{LVR}}\!,
\end{equation} 
which was first observed by Milionis \emph{et al.}\cite{Milionis2023-poisson} in the special case of Poisson-distributed block times. In this factorization, $P_{\text{trade}} = \frac{1}{\mathbb{E}[\tau_1]}$ represents the long-run probability that a given block contains an arbitrage trade (the arbitrage frequency per block), and $\overline{\mathrm{LVR}} = \mathbb{E}\!\big[\,LVR_{\tau_1}\big]$ denotes the average loss per arbitrage event. This decomposition is practical because it separates the frequency of arbitrage opportunities from their average magnitude, allowing each component to be analyzed independently.
\subsection{Relating LVR to Ladder Heights and Overshoot Variables}

\noindent
Interpreting price updates as a random walk allows us to tap into results from renewal theory and fluctuation theory. Central to these is the concept of \emph{ladder heights}, which quantify the extent of new extrema in a random walk. We begin by defining this formally, as it will play a key role in our analysis:

\begin{definition}[Ladder Height]
	Let $(X_n)$ be a sequence of i.i.d. random variables with finite 
	variance. For $k > 0$, let $S_k := X_1 + \cdots + X_k$ be the partial sums 
	associated with the $X_i$'s, and let:
	\[
	\tau_0 := \inf\{k : S_k < 0\}
	\]
	be the first passage time to the negative half-line. The random 
	variable:
	\[
	H := -S_{\tau_0}
	\]
	is the ladder height associated with $(X_n)$.
\end{definition}

In simpler terms, the ladder height $H$ is the absolute value of the random walk's partial sum when it first becomes negative. Let $h_1$ and $h_2$ denote the first and second moments of $H$, respectively. 
We now introduce additional random variables to facilitate our calculations and to leverage well-established results on ladder heights. We embed the Markov chain in a natural i.i.d. random walk $(X_n)_{n \in \mathbb{N}}$ and define $\tau$ as the first step at which an arbitrage occurs, i.e.,
\begin{equation*}
	\tau := \inf\{n > 0: X_n \notin (0, \rho_b)\}.
\end{equation*}

Additionally, define the stopping times:
\begin{align*}
	\tau_0 &:= \inf\{n > 0: X_n < 0\}, \\
	\tau_{\rho_b} &:= \inf\{n > 0: X_n > \rho_b\}.
\end{align*}

In particular, $P_{\mathrm{trade}} = \frac{1}{\mathbb{E}[\tau]}$. Furthermore, we can express the expected LVR as:
\begin{align}\label{eq:lvr_basic}
\overline{\mathrm{LVR}} = \frac{\ell}{2}\mathbb{E}\!\left[ S_{\tau}^2 
\mathbf{1}_{\{\tau=\tau_0\}} + (S_{\tau}-\rho_b)^2 
\mathbf{1}_{\{\tau=\tau_{\rho_b}\}} \right].
\end{align}

Let $p = \mathbb{P}(\tau=\tau_{\rho_b})$ be the probability that the random walk hits the upper bound $\rho_b$ before dropping below $0$. We introduce the following auxiliary random variables:

\begin{itemize}
	\item $H$, the ladder height of the full random walk, is $-S_{\tau_0}$; it has first moment $h_1$ and second moment $h_2$.
	\item The right overshoot $R$ is defined as $R = \frac{(S_{\tau} - \rho_b)\mathbf{1}_{\{\tau=\tau_{\rho_b}\}}}{p}$; it has first moment $r_1$ and second moment $r_2$.
	\item The left overshoot $L$ is defined as $L = \frac{-S_{\tau}\mathbf{1}_{\{\tau=\tau_0\}}}{\,1-p\,}$; it has first moment $l_1$ and second moment $l_2$.
	\item The leftover left overshoot $O$ is defined as $O = \frac{-S_{\tau}\mathbf{1}_{\{\tau=\tau_{\rho_b}\}}}{p}$; it has first moment $o_1$ and second moment $o_2$.
\end{itemize}

$L$ and $R$ capture the overshoot distribution when the chain exits the interval $(0,\rho_b)$ through the left and right side, respectively, while $O$ captures the difference between the distributions of $H$ and $L$ (note that $L$ is used only as an intermediate quantity and will not appear in the final results). 
From these definitions, we derive the following relationships:
\begin{align}
	h_1 &= l_1(1-p) + o_1 p \label{eq:h1_def} \\
	h_2 &= l_2(1-p) + o_2 p \label{eq:h2_def} 
\end{align}

By Wald's identities (classical results in Markov chain theory~\cite{Wald1944}), we have:
\begin{align*}
	\mathbb{E}[X_1]\,\mathbb{E}[\tau] &= \mathbb{E}[S_\tau] \\
	\mathrm{Var}(X_1)\,\mathbb{E}[\tau] + (\mathbb{E}[X_1])^2\,\mathbb{E}[\tau^2] &= \mathbb{E}[S_\tau^2] 
\end{align*}
In our setting, $\mathbb{E}[X_1] = \mathbb{E}[\mathbb{E}[X_1 \mid U_1]] = 0$ regardless of $\mu$, since $X_1$ has zero drift. Using the law of total variance, we also get:
\begin{align*}
	\mathrm{Var}(X_1) &= \mathbb{E}[\mathrm{Var}(X_1 \mid U_1)] + \mathrm{Var}(\mathbb{E}[X_1 \mid U_1]) \\
	&= \mathbb{E}[U_1] + 0 = 1
\end{align*}
Hence, the two Wald's identities simplify to:
\begin{align}
	0 &= \mathbb{E}[S_\tau], \tag{$*$}\label{eq:zero_rel_pre}\\
	\mathbb{E}[\tau] &= \mathbb{E}[S_\tau^2]. \tag{$**$}\label{eq:Etau_pre}
\end{align}
Rewriting these equations in terms of the overshoot variables yields:
\begin{align}
	0 &= -\,l_1(1-p) + (\rho_b + r_1)p, \label{eq:zero_rel}\\
	\mathbb{E}[\tau] &= l_2(1-p) + (\rho_b^2 + 2\rho_b r_1 + r_2)p. \label{eq:Etau_def}
\end{align}

We now solve the system of four equations \eqref{eq:h1_def}, \eqref{eq:h2_def}, \eqref{eq:zero_rel}, and \eqref{eq:Etau_def}. From \eqref{eq:zero_rel}, we obtain $l_1(1-p) = (\rho_b + r_1)p$. Substituting this into \eqref{eq:h1_def} yields:
\begin{align*}
	h_1 &= (\rho_b+r_1)p + o_1 p \\
	&= (\rho_b + r_1 + o_1)p~,
\end{align*}
so solving for $p$ gives:
\begin{equation}
	p = \frac{h_1}{\rho_b + r_1 + o_1} \label{eq:p_solution}
\end{equation}
Next, we express $\mathbb{E}[\tau]$ in terms of $h_2$ and $p$. Substituting $l_2(1-p)$ from \eqref{eq:h2_def} into \eqref{eq:Etau_def} yields:
\begin{align}
	\mathbb{E}[\tau] &= (h_2 - o_2 p) + (\rho_b^2 + 2\rho_b r_1 + r_2)p \nonumber \\
	&= h_2 + (\rho_b^2 + 2\rho_b r_1 + r_2 - o_2)p \label{eq:Etau_intermediate}
\end{align}
Finally, substituting \eqref{eq:p_solution} into \eqref{eq:Etau_intermediate} gives:
\begin{equation}\label{eq:main_tau}
	\mathbb{E}[\tau] \;=\; h_2 \;+\; h_1\, 
	\frac{\rho_b^2 + 2\,\rho_b\,r_1 + r_2 - o_2}{\,\rho_b + r_1 + o_1\,}\,.
\end{equation}

This gives us an expression for $P_{\mathrm{trade}}$ that is independent of $p$, $l_1$, and $l_2$.

In the Poisson block-time case, the memoryless property of the exponential distribution implies that $H$, $R$, and $O$ are all exponentially distributed with parameter $\sqrt{2}$. Consequently, $h_1 = r_1 = o_1 = 1/\sqrt{2}$ and $h_2 = r_2 = o_2 = 1$. Under these conditions, equation~\eqref{eq:main_tau} simplifies to $\mathbb{E}[\tau] = 1 + \frac{\rho_b}{\sqrt{2}}$, and hence $P_{\mathrm{trade}} = \frac{1}{\,1 + \frac{\rho_b}{\sqrt{2}}\,}$. (This closed-form was first noted in \cite{Milionis2023-poisson}.)

Next, we express $\overline{\mathrm{LVR}}$ using the quantities defined above. 
Rewriting equation \eqref{eq:lvr_basic}, $\overline{\mathrm{LVR}}$ can be expressed in a simplified form as:
\begin{equation*}
	\overline{\mathrm{LVR}} = \ell \sigma_b^2 \frac{\mathbb{E}[L^2](1-p) + \mathbb{E}[R^2]p}{2} = \ell \sigma_b^2 \frac{l_2(1-p) + r_2p}{2}
\end{equation*}

Combining this with \eqref{eq:h2_def}, we obtain:
\begin{align}
	\overline{\mathrm{LVR}} &= \frac{\ell \sigma_b^2 \left[h_2 + (r_2 - o_2)p\right]}{2} \nonumber \\
	&= \ell \sigma_b^2 \left[\frac{h_2}{2} + \frac{h_1(r_2 - o_2)}{2(\rho_b + r_1 + o_1)}\right] \label{eq:exp_arb}
\end{align}

Together, these results provide a powerful framework for analyzing LVR if we are able to understand well enough overshoots statistics. In the following sections, we will apply known results on random walks on strip as well as new results to derive closed-form expressions for $P_{\mathrm{trade}}$, $\overline{\mathrm{LVR}}$, and $\overline{\mathrm{ARB}}$.

\section{LVR and Probability of Trade in the constant block-time setting}

In this section, we shift our focus to the constant block-time case, where blocks are produced at fixed, deterministic intervals. By employing concentration results on overshoot distributions derived from Wiener-Hopf factorization, we obtain closed-form approximations that are exponentially accurate in the intra-block volatility. These results offer surprisingly practical and accurate formulas for real-world LVR analysis.

\begin{lemma}[Convergence of Overshoot Variables]\label{lem:Lotov}
	Consider the constants 
	\[
	\kappa = \frac{\lvert\zeta(1/2)\rvert}{\sqrt{2\pi}}, \qquad 
	\omega = \frac{1}{4} + \kappa^2\,.
	\] 
It follows from the work of Lotov~\cite{Lotov1996} that for a constant 
block-time distribution, there exists $c>0$ such that: 
\[
o_1 = \kappa + O(e^{-c\rho_b}), \qquad r_1 = \kappa + O(e^{-c\rho_b})\,,
\] 
\[
o_2 = \omega + O(e^{-c\rho_b}), \qquad r_2 = \omega + O(e^{-c\rho_b})\,.
\] 
\end{lemma}

Although this result is a relatively direct consequence of Theorem 4 in \cite{Lotov1996}, the details of the proof are somewhat technical and are deferred to the appendix for the sake of readability.

This Lemma allows us to derive approximations for LVR in the constant block-time case that are exponentially close to the true value.

\begin{corollary}[LVR in the constant block-time Case]
There exists 
a constant $c>0$ such that:
\[
P_{\mathrm{trade}} \;=\; \frac{1}{\,\frac{\gamma}{\sqrt{2}\,\sigma_b} + 
	\frac{|\zeta(1/2)|}{\sqrt{\pi}}\,} + O\!\Big(e^{-\,c\,\frac{\gamma}{\sigma_b}}\Big)\,,
\] 
\[
\overline{\mathrm{LVR}} \;=\; \frac{\ell\,\sigma_b^2\,|\zeta(1/2)|}{2\sqrt{\pi}} + 
O\!\Big(e^{-\,c\,\frac{\gamma}{\sigma_b}}\Big)\,,
\] 
\[
\overline{\mathrm{ARB}} \;=\; \frac{\ell\,\sigma_b^2}{\,2 + 
	\frac{\sqrt{2\pi}\,\gamma}{|\zeta(1/2)|\,\sigma_b}\,} + 
O\!\Big(e^{-\,c\,\frac{\gamma}{\sigma_b}}\Big) 
\] 
\end{corollary}

\begin{proof}
	For a random walk with increments $X_n \sim N(0,1)$, it is a classical result in probability theory (see, e.g., \cite{Lai1976}) that the ladder height has first two moments 
	\[
	h_1 = \frac{1}{\sqrt{2}}, \qquad h_2 = \,\frac{\lvert\zeta(1/2)\rvert}{\sqrt{\pi}}\,.
	\] 
	(References often use the constant $K = \lim_{n\to\infty}\sum_{k=1}^n \frac{1}{\sqrt{k}} - 2\sqrt{n}$ in place of $\lvert\zeta(1/2)\rvert$; one can show that these are equal by applying the Euler--Maclaurin formula to the Riemann zeta function.)
	
	The expected value of $\tau$ is 
	\[
	\mathbb{E}[\tau] = h_2 + h_1\,\frac{\rho_b^2 + 2\,\rho_b\,r_1 + r_2 - o_2}{\rho_b + r_1 + o_1}\,. 
	\] 
	Recall that $\rho_b=\frac{\gamma}{\sigma_b}$. By Lemma~\ref{lem:Lotov}, we have 
	\[
	r_1 = \kappa + O(e^{-c\rho_b}), \qquad o_1 = \kappa + O(e^{-c\rho_b})\,,
	\] 
	\[
	r_2 = \omega + O(e^{-c\rho_b}), \qquad o_2 = \omega + O(e^{-c\rho_b})\,.
	\] 
	The denominator simplifies to:
	\[
	\rho_b + r_1 + o_1 = \rho_b + 2\kappa + O(e^{-c\rho_b})\,.
	\] 
	For the numerator, we note that 
	\[
	r_2 - o_2 = O(e^{-c\rho_b})\,,
	\] 
	and hence 
	\[
	\rho_b^2 + 2\,\rho_b\,r_1 + r_2 - o_2 = \rho_b\big(\rho_b + 2\kappa \big) + O(e^{-c\rho_b})\,.
	\] 
	Thus, the ratio simplifies to 
	\[
	\frac{\rho_b^2 + 2\,\rho_b\,r_1 + r_2 - o_2}{\,\rho_b + r_1 + o_1\,} 
	= \frac{\rho_b\big(\rho_b + 2\kappa \big) + O(e^{-c\rho_b})}{\,\rho_b + 2\kappa  + O(e^{-c\rho_b})\,} 
	= \rho_b + O(e^{-c\rho_b})\,.
	\] 
	Hence, the expression for $\mathbb{E}[\tau]$ simplifies to 
	\begin{align*}
		\mathbb{E}[\tau] &= \,\frac{\lvert\zeta(1/2)\rvert}{\sqrt{\pi}} + \frac{1}{\sqrt{2}}\Big(\rho_b + O(e^{-c\rho_b})\Big) \\
		&= \,\frac{\lvert\zeta(1/2)\rvert}{\sqrt{\pi}} + \frac{\rho_b}{\sqrt{2}} + O(e^{-c\rho_b})\,.
	\end{align*}
	Therefore, 
	\[
	P_{\mathrm{trade}} = \frac{1}{\mathbb{E}[\tau]} = \frac{1}{\frac{\rho_b}{\sqrt{2}} + \frac{\lvert\zeta(1/2)\rvert}{\sqrt{\pi}}} + O(e^{-c\rho_b})\,,
	\] 
	which matches the claimed expression for $P_{\mathrm{trade}}$.
	
	Similarly, 
	\begin{align*}
		\overline{\mathrm{LVR}} &= \ell\,\sigma_b^2 \left[\frac{h_2}{2} + \frac{h_1(r_2 - o_2)}{2(\rho_b + r_1 + o_1)}\right] \\
		&= \ell\,\sigma_b^2 \left[\frac{\,\lvert\zeta(1/2)\rvert}{2\sqrt{\pi}} + \frac{O(e^{-c\rho_b})}{2\big(\rho_b + 2\kappa  + O(e^{-c\rho_b})\big)}\right] \\
		&= \,\frac{\ell\,\sigma_b^2\,\lvert\zeta(1/2)\rvert}{2\sqrt{\pi}} + O(e^{-c\rho_b})\,,
	\end{align*}
	which matches the claimed expression.
	
	Finally, combining the above results with the decomposition identity \eqref{eq:decomp} (i.e., $\overline{\mathrm{ARB}} = P_{\mathrm{trade}} \times \overline{\mathrm{LVR}}$) yields 
	\begin{align*}
		\overline{\mathrm{ARB}} &= P_{\mathrm{trade}} \times \overline{\mathrm{LVR}} \\
		&= \Big( \frac{1}{\frac{\gamma}{\sqrt{2}\,\sigma_b} + \frac{\lvert\zeta(1/2)\rvert}{\sqrt{\pi}}} + O(e^{-c\rho_b}) \Big)\Big( \,\frac{\ell\,\sigma_b^2\,\lvert\zeta(1/2)\rvert}{2\sqrt{\pi}} + O(e^{-c\rho_b}) \Big) \\
		&= \frac{\ell\,\sigma_b^2}{\,2 + \frac{\sqrt{2\pi}\rho_b}{\lvert\zeta(1/2)\rvert}\,} + O(e^{-c\rho_b})\,,
	\end{align*}
	which concludes the proof.
\end{proof}

\begin{remark}
	The actual values of $\kappa$ and $\omega$ are irrelevant for our proof; what matters is only that the overshoot sequences converge exponentially fast to the same finite limits. Therefore, one could adapt the proof in \cite{Lotov1996} to other choices of $\mu$. However, the exact arguments used in that paper do not generalize directly to an arbitrary distribution $\mu$, since they rely on certain analytic properties of the characteristic function of the normal distribution that are not always satisfied by other distributions. For example, the proof would not work as is for Poisson distributed blocks. 
\end{remark}

\section{Results on general distributions}

In this section we broaden the scope of our analysis by extending it to arbitrary block-time distributions, aiming to uncover universal properties of LVR and assess their implications across diverse blockchain architectures. Using the same Markov chain framework, we will show that the asymptotic arbitrage probability in a block remains consistent across distributions to first order. On the other hand, the magnitude of LVR varies with the block-time distribution. A key finding is that the Dirac distribution uniquely minimizes the asymptotic LVR among all distributions with a fixed mean.

\begin{lemma}\label{lemma:general_bound}
	Let $u_2$ be the second moment of $\mu$. Then $ 0 \leq r_1,\,o_1 \leq \sqrt{\frac{2}{\pi}}$ and $ 0 \leq r_2,\,o_2 \leq u_2$.
\end{lemma}

As with Lemma~\ref{lem:Lotov}, the proof of this result is relatively natural and is deferred to the Appendix for clarity.

We can now prove the following two important corollaries regarding $LVR$ for general distributions.

\begin{corollary}[Distribution-Independent Arbitrage Probability]\label{cor:distr_indep}
	To first order in $\sigma_b$ (as $\sigma_b \to 0$), the asymptotic arbitrage 
	probability $P_{\mathrm{trade}}$ is independent of the block-time distribution $\mu$. 
	Specifically, 
	\[
	P_{\mathrm{trade}} = \frac{\sqrt{2}\,\sigma_b}{\gamma} + O(\sigma_b/\gamma)\,. 
	\] 
\end{corollary}

\begin{proof}
From Spitzer’s result \cite{Spitzer1956}, we know that for a centered random walk, the expected ladder height is 
\[ h_1 = \frac{\sqrt{\Var(X)}}{\sqrt{2}}\,. \] 
We proved in Section~2 that $\Var(X_1) = 1$ (independent of $\mu$), thus $h_1 = 1/\sqrt{2}$. 
Now, applying equation~\eqref{eq:main_tau}, we have 
\begin{align*}
	\mathbb{E}[\tau] &= h_2 \;+\; \frac{1}{\sqrt{2}} \,
	\frac{\rho_b^2 + 2\,\rho_b\,r_1 + r_2 - o_2}{\,\rho_b + r_1 + o_1\,} \\
	&= h_2 \;+\; \frac{\rho_b}{\sqrt{2}} \Bigg(1 + \frac{2\,\rho_b\,r_1 + r_2 - o_2 - o_1 - r_1}
	{\rho_b^2 + \rho_b\,r_1 + \rho_b\,o_1}\Bigg)\,. 
\end{align*}
Since, by Lemma~\ref{lemma:general_bound}, the quantities $r_1, r_2, o_1, o_2,$ and $h_2$ 
are all bounded constants (independent of $\rho_b$), the fraction in the parentheses 
is $O(1/\rho_b)$. Consequently, $\mathbb{E}[\tau] = \frac{\rho_b}{\sqrt{2}} + O(1/\rho_b)$. 
This concludes the proof.
\end{proof}

\begin{corollary}[Optimality of the Dirac distribution]
	The asymptotic expected $LVR$ per block and the expected arbitrage per block are, respectively:
\begin{equation}\label{eq:generallvr}
	\overline{\mathrm{LVR}} \;=\; 
	\frac{\ell\,\sigma_b^2}{2}\Bigg(\frac{|\zeta(1/2)|}{\sqrt{2\pi}} + C_\mu\Bigg) + O(\sigma_b/\gamma)\,. 
\end{equation}
\begin{equation}\label{eq:generalarb}
	\overline{\mathrm{ARB}} \;=\; 
	\frac{\ell\,\sigma_b^3}{\sqrt{2}\,\gamma}\Bigg(\frac{|\zeta(1/2)|}{\sqrt{2\pi}} + C_\mu\Bigg) + O(\sigma_b/\gamma)\,. 
\end{equation}
	where $C_\mu$ is a non-negative constant depending only on $\mu$. As an important consequence, asymptotically the Dirac distribution achieves the lowest possible expected $LVR$ and expected arbitrage among all block-time distributions~$\mu$. 
\end{corollary}

\begin{proof}
	Equation~\eqref{eq:generalarb} follows directly from Corollary~\ref{cor:distr_indep} combined with Equation~\eqref{eq:generallvr}. We thus restrict our attention to proving Equation~\eqref{eq:generallvr}.
	
	Similarly as in the previous Corollary, since, by Lemma~\ref{lemma:general_bound}, the quantities $r_1, r_2, o_1, o_2,h_1$ and $h_2$ 
	are all bounded constants, the fraction $\frac{h_1(r_2 - o_2)}{2(\rho_b + r_1 + o_1)}$ in \eqref{eq:exp_arb} is $ O(1/\rho_b)$ and the equation can be rewritten as
	\begin{align}
		\overline{\mathrm{LVR}} &= \frac{\ell \sigma_b^2h_2}{2} + O(1/\rho_b) \label{eq:exp_arb_general}
	\end{align}
	Let us now focus on the $h_2$ term. As derived by Lai in \cite{Lai1976}, the second moment of the ladder height for a random walk can be expressed as
	\begin{equation}\label{eq:h2}
		h_2 = \Bigg(\frac{\lvert\zeta(1/2)\rvert}{\sqrt{2\pi}}+ \frac{1}{3\sqrt{2}}\,\mathbb{E}[X_1^3] - \frac{1}{\sqrt{2}}\sum_{n=1}^{\infty}\frac{1}{\sqrt{n}} \Bigg(\mathbb{E}\Big[\frac{S_n^-}{\sqrt{n}}\Big] - \frac{1}{\sqrt{2\pi}}\Bigg)\Bigg)\,e^{\,h_1 - \frac{1}{\sqrt{2}}}\,.
	\end{equation}
	We now analyze each term in Equation~\eqref{eq:h2}. For all $\mu$, $\mathbb{E}[X_1^3] = 0$ since $X_1$ is symmetric (the underlying normal random variable has mean zero in our setting). Additionally, from the proof of Corollary~\ref{cor:distr_indep} we know $e^{\,h_1 - 1/\sqrt{2}} = 1$, independently of $\mu$, because the first ladder height moment does not depend on the block distribution. Hence, the only remaining contribution to $h_2$ comes from the infinite sum term.
	
	%In the uniform block-time case, $|S_n^-| = |\sqrt{n}\,\mathcal{N}(0,1)|$ is a folded normal random variable with mean $\frac{\sqrt{n}}{2\pi}$; hence, the infinite sum vanishes. In the general case, $|S_n^-| = \big| \mathcal{N}\!\Big(0, \sum_{i=1}^{n}U_i\Big) \big|$, whose expectation is $\frac{\mathbb{E}\!\Big[\,(\sum_{i=1}^{n}U_i)^{1/2}\Big]}{2\pi}$. Although there is no closed-form expression for this expectation, we can apply Jensen’s inequality to obtain 
	%\[
	%\mathbb{E}\!\Bigg[\Big(\sum_{i=1}^{n}U_i\Big)^{1/2}\Bigg] \;\le\; \Big(\mathbb{E}\big[\sum_{i=1}^{n}U_i\big]\Big)^{1/2} \;=\; \sqrt{n}\,. 
	%\] 
	%As a result, the constant $C_\mu$ (introduced above) is nonnegative for any choice of $\mu$, and it is minimized (equal to zero) only in the degenerate case where $U$ is a Dirac distribution.
	
	%Combining all these observations with Equation~\eqref{eq:exp_arb}, we obtain 
	%\[
	%\overline{\mathrm{LVR}}_{\text{Uniform, no costs}} = \frac{\ell\,\sigma_b^2}{2}\,\Big(\frac{\lvert\zeta(1/2)\rvert}{\sqrt{2\pi}} + C_\mu\Big) + o(\sigma_b)\,,
	%\] 
	%where 
	%$$
	%C_\mu = \frac{1}{\sqrt{2}\,\pi}\sum_{n=1}^{\infty} \left(\frac{1}{\sqrt{n}} \mathbb{E}\Big[ \sqrt{\sum_{i=1}^{n}U_i}\Big] - 1\right)\, \geq 0. 
	%$$ 
	%In particular, the expected $LVR$ is minimized when the block-time distribution is uniform.
	In our setting, the $S_i$'s are symmetric independently of the distribution hence $S_n^- =\frac{|S_n|}{2}$. In the constant block-time case, $|S_n| = |\sqrt{n}\,\mathcal{N}(0,1)|$ is a folded 
	normal random variable with mean $\displaystyle\sqrt{\frac{2n}{\pi}}$ and thus $\mathbb{E}[\,|S_n^-|/\left(2\sqrt{n}\right)\,] = \frac{1}{\sqrt{2\pi}}$ . Hence, the infinite 
	sum term in \eqref{eq:h2} vanishes.
	
	In the general case, $|S_n| = \Big|\mathcal{N}\!\Big(0, \sum_{i=1}^n U_i\Big)\Big|$, whose 
	expectation is $\displaystyle\sqrt{\frac{2}{\pi}}\;\mathbb{E}\!\Big[\sqrt{\sum_{i=1}^n U_i}\,\Big]$.
	Although there is no closed-form expression for this expectation in general, it is necessarily smaller
	than $\displaystyle\sqrt{\frac{2n}{\pi}}$ due to Jensen’s inequality. Therefore, $C_\mu= - \frac{1}{\sqrt{2}}\sum_{n=1}^{\infty}\frac{1}{\sqrt{n}} \Bigg(\mathbb{E}\Big[\frac{S_n^-}{\sqrt{n}}\Big] - \frac{1}{\sqrt{2\pi}}\Bigg)$ defined 
	in \eqref{eq:generallvr} is non-negative, and in fact $C_\mu=0$ if and only if $\mu$ is a Dirac distribution. This 
	confirms that asymptotically the Dirac distribution achieves the lowest possible $\overline{\mathrm{LVR}}$ and  $\overline{\mathrm{ARB}}$.
	
\end{proof}

\section{Final remarks}

We conclude with a few additional observations and potential directions for future work:

\begin{enumerate}
	\item When $\sigma_b \gg \sigma$, the process is well-approximated by a geometric distribution with parameter $\frac{\sigma}{\sqrt{2\pi}\,\sigma_b}$. Consequently, our asymptotic approximation (which assumes small enough $\sigma_b$) breaks down in this extreme regime.
	\item To better align our model with real-world conditions, it is necessary to incorporate the presence of fees both for arbitrageurs and liquidity providers. Fortunately, the influence of fees can be studied by shifting the starting point of the Markov chain or expanding the interval. Incorporating fees in this manner is an important extension that we leave for future work.
	\item One can define an LVR process for a given CFAMM (constant-function AMM), where $\ell(x)$ denotes the loss-versus-rebalancing at state $x$. In this setting, we obtain the natural SDE for LVR as:
	\[ d\mathrm{LVR}_t = \frac{\ell(X_t)\,\sigma^{2}}{\,2+\sqrt{2\pi}\,\gamma/(|\zeta(1/2)|\,\sigma\sqrt{T})\,}\,dt~, \] 
	where $\sigma$ is the asset volatility and $T$ is the average block interval. This opens a new perspective for analyzing the value of an LP’s portfolio on a macroscopic scale.
\end{enumerate}

\section{Appendix}

\begin{proof}[Proof of Lemma \ref{lem:Lotov}]
	Let $\{X_i\}_{i\ge1}$ be i.i.d.\ $\mathcal{N}(0,1)$ and define $S_n=\sum_{i=1}^n X_i$ with $S_0=x$. Let $\tau_{\rho_b}=\inf\{n\ge1:S_n\ge\rho_b\}$ be the hitting time of the upper barrier $\rho_b$, and $\tau_0=\inf\{n\ge1:S_n<0\}$ the hitting time of~0. By Theorem 4 of Lotov \cite{Lotov1996}, there exists a constant $c>0$ such that for all $x\ge0$ and large $\rho_b$:
	\begin{equation}\label{eq:lotov-overshoot}
		\begin{aligned}
			\Bigl|\Prob\!\bigl(S_{\tau_{\rho_b}}>\rho_b+x \mid \tau_{\rho_b}<\tau_0\bigr)-F(x)\Bigr|
			&\le e^{-\,c\,\min\{x,\rho_b-x\}},\\[6pt]
			\Bigl|\Prob\!\bigl(S_{\tau_0}<x \mid \tau_0<\tau_{\rho_b}\bigr)-F(x)\Bigr|
			&\le e^{-\,c\,\min\{x,\rho_b-x\}},
		\end{aligned}
	\end{equation}
	where $F(x)$ is a distribution function with mean $\kappa$ and second moment $\omega$ as stated. 
	
	To make the error uniform in $x$, restrict to high-probability “good’’ events where the walk stays well away from the boundaries before finishing its crossing. For the upward-crossing case define
	\[
	B:=\Bigl\{\tau_{\rho_b}<\tau_0\text{ and }\exists n_0\ge\tau_{\rho_b}\!:\;
	S_{n_0}\in[\rho_b/4,\,3\rho_b/4]\Bigr\}.
	\]
	On $B$ we may restart the walk at~$n_0$ (strong Markov property); since the distance to either barrier is then at least $\rho_b/4$, \eqref{eq:lotov-overshoot} gives
	\begin{equation}\label{eq:uniform-error-B}
		\Bigl|\Prob\!\bigl(S_{\tau_{\rho_b}}>\rho_b+x \mid B\bigr)-F(x)\Bigr|
		\le e^{-\,c\,\rho_b/4}\quad\text{for all }x\ge0.
	\end{equation}
	For the downward-crossing case set
	\[
	A:=\Bigl\{\tau_0<\tau_{\rho_b}\text{ and }\exists n_0\ge\tau_0\!:\;
	S_{n_0}\in[\rho_b/4,\,3\rho_b/4]\Bigr\},
	\]
	and analogously obtain
	\begin{equation}\label{eq:uniform-error-A}
		\Bigl|\Prob\!\bigl(S_{\tau_0}<x \mid A\bigr)-F(x)\Bigr|
		\le e^{-\,c\,\rho_b/4}\quad\text{for all }x\ge0.
	\end{equation}
	
We next show that $A^{c}$ and $B^{c}$ are negligible.  
Either event can happen only if the walk jumps from one barrier to the other
without first visiting the strip $[\rho_{b}/4,\,3\rho_{b}/4]$,
which forces a single increment of size at least $\rho_{b}/2$.
Set
\[
\tau_A \;:=\; \inf\{i\ge 1:\,|X_i|\ge \rho_{b}/2\}.
\]
For every $m\ge 1$,
\[
\Prob(\tau_A\le m)\;\le\; m\,e^{-c'\rho_{b}^{2}}.
\]
In addition, standard results for symmetric walks give 
$\Prob(\tau_0>n)=O(n^{-1/2})$.
Choosing $m:=\lceil e^{\rho_{b}^{2}/4}\rceil$ yields
\[
\Prob(\tau_A<\tau_0)
\;\le\;
\Prob\!\bigl(\tau_A\le e^{\rho_{b}^{2}/4}\bigr)
+\Prob\!\bigl(\tau_0>e^{\rho_{b}^{2}/4}\bigr)
=O\!\bigl(e^{-\rho_{b}^{2}/4}\bigr).
\]
Hence there are constants $C,c'>0$ such that
\begin{equation}\label{eq:skip-prob-bound}
	\Prob\!\Bigl(\exists\,i<\tau_0:|X_i|\ge\rho_{b}/2\Bigr)
	\;\le\; C\,e^{-c'\rho_{b}^{2}},
\end{equation}
and therefore $\Prob(A^{c}),\Prob(B^{c})\le C\,e^{-c'\rho_{b}^{2}}$.

	Combine \eqref{eq:uniform-error-B}, \eqref{eq:uniform-error-A} and \eqref{eq:skip-prob-bound}. For any $x\ge0$,
	\[
	\Bigl|\Prob\!\bigl(S_{\tau_{\rho_b}}>\rho_b+x \mid \tau_{\rho_b}<\tau_0\bigr)-F(x)\Bigr|
	\le e^{-\,c\,\rho_b/4}+C e^{-\,c'\rho_b^2}=O(e^{-c\rho_b}),
	\]
	and similarly for the left overshoot. Uniform convergence of distributions implies convergence of moments: using $\E[Z]=\int_0^\infty\Prob(Z>x)\,dx$ and $\E[Z^2]=2\int_0^\infty x\,\Prob(Z>x)\,dx$,
	\[
	r_1=\int_0^\infty\!\Prob(S_{\tau_{\rho_b}}>\rho_b+x\mid \tau_{\rho_b}<\tau_0)\,dx
	=\int_0^\infty(1-F(x))\,dx+O(e^{-c\rho_b})=\kappa+O(e^{-c\rho_b}),
	\]
	and likewise $o_1=\kappa+O(e^{-c\rho_b})$, $r_2=\omega+O(e^{-c\rho_b})$, $o_2=\omega+O(e^{-c\rho_b})$. This completes the proof of Lemma~\ref{lem:Lotov}.
\end{proof}

\begin{proof}[Proof of Lemma \ref{lemma:general_bound}]
	The lemma follows from the fact that both \(R\) and \(O\) are stochastically dominated by a folded normal distribution \( |\mathcal{N}(0,\mu)| \), which has first moment \( \sqrt{\frac{2}{\pi}} \) and second moment \( u_2 \). We demonstrate this for the random variable \(R\); the argument for \(O\) is analogous. 
	
	Let \( (B_t)_{t \ge 0} \) be a standard Brownian motion, and consider an i.i.d. sequence \( (U_k)_{k \ge 0} \) of random variables drawn from the distribution $\mu$ (representing the successive block times). Define the cumulative times \(T_i = \sum_{k=1}^i U_k\) (with \(T_0 = 0\)). Then \(X_i := B_{T_i} - B_{T_{i-1}}\) for \(i \ge 1\) defines a random walk. By construction, $B_{T_i}$ is the position of the Brownian motion at time $T_i$, so the sequence $(X_i)$ has the same distribution as our Markov chain.
	
	In particular,
	\[
	\mathbb{P}(R_1 \ge x \mid \tau_0 \ge \tau_{\rho_b}) \;=\; \mathbb{P}(X_{\tau} \ge b + x \mid \tau_0 \ge \tau_{\rho_b})\,.
	\]
	Using the law of total probability, we condition on the number of steps $k$ taken until the stopping time $\tau$:
	\[
	\mathbb{P}(R_1 \ge x \mid \tau_0 \ge \tau_{\rho_b}) \;=\; \sum_{k=1}^\infty \mathbb{P}(X_{\tau} \ge b + x \mid \tau_0 \ge \tau_{\rho_b},\, \tau = k)\, \mathbb{P}(\tau = k \mid \tau_0 \ge \tau_{\rho_b})\,.
	\]
	For each $k$, on the event $\{\tau = k\}$ the position at the stopping time is $X_{\tau} = X_k = B_{T_k} - B_{T_{k-1}}$, where $T_k = \sum_{i=1}^k U_i$. In particular, given $\tau = k$, the Brownian motion runs for an additional time $u = U_k$ after time $T_{k-1} = \sum_{i=1}^{k-1} U_i$. Thus,
\begin{align*}
	&\mathbb{P}(X_{\tau} \ge b + x \mid \tau_0 \ge \tau_{\rho_b},\, \tau = k) = \\
	&\int_{u \sim \mu} \int_\mu \cdots \int_\mu
	\mathbb{P}\!\Big(
	B_{T_{k-1}+ u} \ge b + x
	\Big)
	\mathbb{P}\left(0\leq B_{T_1},...,B_{T_{k-1}}\leq b\right)dP_{U_1}(u_1)\, \cdots\, dP_{U_{k-1}}(u_{k-1})\, d\mu(u)\,,
\end{align*}

	where $dP_{U_i}(u_i)$ denotes the distribution measure of $U_i$ and the outer integration is over $u$ distributed according to $\mu$ (i.e. $u = U_k$).
	
	Dividing both sides by $\mathbb{P}(\tau_0 \ge \tau_{\rho_b})$ to normalize the conditional probability, we obtain
\begin{align*}
	&\mathbb{P}(R_1 \ge x \mid \tau_0 \ge \tau_{\rho_b}) = \\
	&\frac{1}{\mathbb{P}(\tau_0 \ge \tau_{\rho_b})}
	\sum_{k=1}^\infty
	\int_{u \sim \mu}
	\int_\mu \cdots \int_\mu
	\mathbb{P}\!\Big( B_{T_{k-1}+ u} \ge b + x \Big) \times \\
	&\mathbb{P}\left(0\leq B_{T_1},...,B_{T_{k-1}}\leq b\right)dP_{U_1}(u_1)\, \cdots\, dP_{U_{k-1}}(u_{k-1})\, d\mu(u)\,.
\end{align*}

	To bound the integrand, consider 
	\[
	\mathbb{P}\!\Big( B_{T_{k-1}+ u} \ge b + x \Big)\,.
	\] 
	Using properties of Brownian motion, this probability can be bounded by splitting at the hitting time of $b$. Specifically,
\begin{align*}
	&\mathbb{P}\!\Big( B_{T_{k-1}+ u} \ge b + x \Big)
	\le \\
	&\mathbb{P}\!\Big( \sup_{t \in [\,\sum_{i=0}^{k-1} u_i,\; T_{k-1}+ u\,]} B_t \ge b \Big)
	\;\cdot\;
	\mathbb{P}\!\Big( |B_{T_{k-1}+ u} - B_{\tau_{b,k}}| \ge x \,\Big|\, B_{\tau_{b,k}} = b \Big)\,,
\end{align*}

	where $\tau_{b,k}$ denotes the first time the Brownian motion hits level $b$ during the interval $\big[\sum_{i=0}^{k-1} u_i,\; T_{k-1}+ u\big]$. The term $|B_{T_{k-1}+ u} - B_{\tau_{b,k}}|$ represents the overshoot beyond $b$. Since $\tau_{b,k} \ge \sum_{i=0}^{k-1} u_i$ by definition (the process cannot hit $b$ before the interval starts), the remaining time after hitting $b$ is at most $u$. Therefore, the overshoot is stochastically dominated by a folded normal with variance $u$. In particular,
	\[
	\mathbb{P}\!\Big(|B_{T_{k-1}+ u} - b| \ge x \Big) \;\le\; \mathbb{P}\!\big(|\mathcal{N}(0, u)| \ge x\big)\,.
	\]
	
	Substituting this bound into the previous expression, we get
\begin{align*}
	&\mathbb{P}(R_1 \ge x \mid \tau_0 \ge \tau_{\rho_b}) \le \\[2pt]
	&\frac{1}{\mathbb{P}(\tau_0 \ge \tau_{\rho_b})}
	\sum_{k=1}^\infty
	\int_{u \sim \mu}
	\mathbb{P}\bigl(|\mathcal{N}(0,u)| \ge x\bigr)\, \times \\
	&\Bigl(
	\int_\mu \!\cdots\! \int_\mu\mathbb{P}\left(0\leq B_{T_1},...,B_{T_{k-1}}\leq b\right)
	dP_{U_1}(u_1)\,\cdots\,dP_{U_{k-1}}(u_{k-1})
	\Bigr)\,
	d\mu(u)\,.
\end{align*}

	The inner $(k-1)$ integrals together with the summation $\sum_{k=1}^\infty$ yield $\mathbb{P}(\tau_0 \ge \tau_{\rho_b})$, canceling out the normalizing factor. Thus, 
	\[
	\mathbb{P}(R_1 \ge x \mid \tau_0 \ge \tau_{\rho_b}) \;\le\; \int_{u \sim \mu} \mathbb{P}(|\mathcal{N}(0, u)| \ge x)\, d\mu(u) \;=\; \mathbb{P}(|\mathcal{N}(0,\mu)| \ge x)\,. 
	\]
	
	This shows that \(R_1\) is stochastically dominated by \(|\mathcal{N}(0,\mu)|\). By symmetry, the same argument applies to \(R_1\). Since the folded normal \( |\mathcal{N}(0,\mu)| \) has first moment \( \sqrt{\frac{2}{\pi}} \) and second moment \( u_2 \), it follows that 
	\[ 
	r_1, o_1 \le \sqrt{\frac{2}{\pi}}, \qquad r_2, o_2 \le u_2\,.
	\]
	This completes the proof.
\end{proof}

\end{document}